\documentclass[11pt]{amsart}

\usepackage[centertags]{amsmath}
\usepackage{amsthm,amsfonts}
\usepackage{amssymb}
\usepackage{epsfig}
\usepackage{amssymb,latexsym}
\usepackage{hyperref}
\usepackage{indentfirst}

\usepackage{mathrsfs}
\usepackage{amsthm}
\usepackage{amsmath}
\usepackage{amssymb}
\usepackage{amsfonts}
\usepackage{amsbsy}
\usepackage{appendix}
\usepackage{bbding}
\usepackage{tikz}
\usetikzlibrary{matrix,fit,shapes.misc,positioning}
\tikzset{%
    highlight/.style={rectangle,rounded corners,fill=red!15,draw,fill opacity=0.3,thick,inner sep=0pt}
}
\tikzset{%
    highlight1/.style={rectangle,rounded corners,fill=blue!15,draw,fill opacity=0.3,thick,inner sep=0pt}
}

\theoremstyle{plain}
\newtheorem{thm}{Theorem}[section]
\newtheorem{cor}[thm]{Corollary}

\theoremstyle{definition}
\newtheorem{defn}[thm]{Definition}
\newtheorem{rem}[thm]{Remark}

\numberwithin{equation}{section}

\newcommand{\E}{\mathbb{E}}
\newcommand{\F}{{\mathbb F}}

\newcommand{\Tr}{{\rm Tr}}

\newcommand{\Hom}{{\rm Hom}}
\newcommand{\K}{{\mathbb K}}
\newcommand{\sF}{{\mathcal F}}
\newcommand{\sE}{{\mathcal E}}

\begin{document}

    \title[The Concatenated Structure of QA Codes] {The Concatenated Structure of Quasi-Abelian Codes}
    \maketitle

    \author
    { {\large \begin{center} Martino Borello$^{1}$, Cem G\"{u}neri$^{2}$, Elif Sa\c{c}\i kara$^{2}$, Patrick Sol\'{e}$^{1}$ \end{center} }
        \vspace{0.3cm}
        \small
        \begin{center}
            $^{1}$  Universit\'e Paris 13, Sorbonne Paris Cit\'e, LAGA, CNRS, UMR 7539, Universit\'e Paris 8, F-93430, Villetaneuse, France\\
            borello@math.univ-paris13.fr, sole@math.univ-paris13.fr\\
            $^{2}$ Sabanc{\i} University, Faculty of Engineering and Natural Sciences, 34956 Istanbul, Turkey \\
            guneri@sabanciuniv.edu, elifsacikara@sabanciuniv.edu
        \end{center}
    }


    \abstract
    The decomposition of a quasi-abelian code into shorter linear codes over larger alphabets was given in (Jitman, Ling, (2015)), extending the analogous Chinese remainder decomposition of quasi-cyclic codes (Ling, Sol\'e, (2001)). We give a concatenated decomposition of quasi-abelian codes and show, as in the quasi-cyclic case, that the two decompositions are equivalent. The concatenated decomposition allows us to give a general minimum distance bound for quasi-abelian codes and to construct some optimal codes. Moreover, we show by examples that the minimum distance bound is sharp in some cases. In addition, examples of large strictly quasi-abelian codes of about a half rate are given. The concatenated structure also enables us to conclude that strictly quasi-abelian linear complementary dual codes over any finite field are asymptotically good. 
    \endabstract

    \vspace{0.3cm}


    \vspace{0.4cm}

    \maketitle

    \section{Introduction}
    The well-known family of quasi-cyclic (QC) codes contains examples of good codes (\cite{GHM, HKLL}) and it is also asymptotically good (\cite{LS1,W}). As shown by Ling and Sol\'{e} in \cite{ls}, a QC code over $\F_q$ (the finite field with $q$ elements, where $q$ is a prime power) can be decomposed into shorter linear codes over various extensions of $\F_q$, using the Chinese Remainder Theorem (so-called CRT decomposition). Moreover, QC codes can also be decomposed into concatenated codes as shown by Jensen (\cite{j}). It was observed in \cite{go} that these two decompositions are equivalent. More specifically, the CRT components (constituents) of a QC code and the outer codes in its concatenated structure are the same.

    The family of quasi-abelian (QA) codes has been introduced by Wasan (\cite{wasan}) in order to generalize the algebraic structure of QC codes. Jitman and Ling reconsidered QA codes (\cite{js}) giving, among other things, the decomposition of QA codes by extending the CRT decomposition for QC introduced by Ling-Sol\'e (see also \cite{dkl}). Let us note that a special class of QA codes is also studied in \cite{goz}, which is called by the authors multidimensional QC codes, or quasi $n$D cyclic codes. In the general case, the family of QA codes coincides with the family QC codes, but if one considers strictly QA codes, these form a proper subfamily of QC codes. Let us also note that QA codes are a special class of codes over a group algebra $\mathbb{F}_{q}[H]$ in $\mathbb{F}_q[G],$ where $H$ is a subgroup of a finite group $G$. In this work, we assume that $G$ is a finite abelian group and $(q,|H|)=1,$ namely $\mathbb{F}_{q}[H]$ is a semisimple algebra. The cases in which $G$ is nonabelian or the algebra is not semisimple are more difficult to be treated in full generality. For some results in this direction, the interested reader may refer to \cite{borello} and references therein.

    Here, we contribute to the structural understanding of QA codes by giving their concatenated decomposition. As in the QC case, we show that the decomposition given by Jitman-Ling and the concatenated decomposition are equivalent. However, there are several advantages of viewing a code in the concatenated form. Firstly, we can transfer the general minimum distance bound for concatenated codes to QA codes. This is, to the best of our knowledge, the first general minimum distance bound on QA codes. Moreover, using {\sc Magma} (\cite{Magma}), we obtain numerical results for the minimum distance of some QA codes, based on their concatenated structure.

    A $q$-ary linear code $\mathcal{C}$ is said to be linear complementary dual (LCD) if  $\mathcal{C} \cap \mathcal{C}^{\bot}=\{0\}$. In \cite{jpc}, Jitman \emph{\emph{et al.}} showed that binary QA LCD codes of index $3$ are asymptotically good by using the asymptotic goodness of binary self-dual QA codes of index $2$, which was proved in \cite{js}. We can show that the family of (strictly) QA LCD codes, over any finite field, are asymptotically good, using the concatenated structure obtained for QA codes.


    The sections are arranged as follows. Introductory material are presented in Section~\ref{background}. The concatenated structure of QA codes is given in Section \ref{results}, which is followed by consequences on the minimum distance. Numerical results and asymptotic conclusions are presented in Section \ref{results2}. 

    \medskip

    \section{Preliminaries}\label{background}

    In this section, we give the background materials used in this work. We assume that the reader is familiar with the definition and general theory of linear codes (\cite{sloane}).

    \subsection{Asymptotics}

    Let $B_{q}(n,d)$ denote the maximum cardinality of a linear code $\mathcal{C},$ for given positive integers $n$ and $d,$ such that a $q$-ary linear code $\mathcal{C}$ with parameters $[n,k,d]$ exists. In general, we have $k:=\log_{q}|\mathcal{C}|,$ where $|\mathcal{C}|$ denotes the cardinality of $\mathcal{C} \subseteq \mathbb{F}_{q}^{n}.$ A linear code $\mathcal{C}$ is said to be \emph{optimal} if it contains exactly $B_{q}(n,d)$ elements.

    If we consider a family of $q$-ary linear codes $\mathcal{C}_{(n)}$ with parameters $[n,k_{n},d_{n}],$ then recall that the rate and the relative distance of the family is defined, respectively, as
    $$r:=\liminf_{n \rightarrow \infty}{k_{n}}/{n},$$
    $$\delta:=\liminf_{n \rightarrow \infty}{d_{n}}/{n}.$$

    A family of $q$-ary linear codes $\mathcal{C}_{(n)}$ is called \emph{asymptotically good} if its rate $r$ and relative distance $\delta$ are both nonzero.

    \subsection{Generalized Concatenated Codes}

    Throughout the paper, by the concatenated structure, we mean generalized concatenated codes (GCC) introduced by Block-Zyablov (\cite{bz}). Since our main results are based on this structure, we find it useful to remind the structure of GCC, following \cite{d}. The idea of this construction is to extend a simple concatenation to a general concatenation with more than one outer code, which are of the same length but defined over possibly different finite extensions of $\mathbb{F}_q$.

    \begin{defn}\label{gcc}
        For $i\in\{1,\ldots,s\},$ let $\mathcal{C}_i$'s be linear codes (called \emph{outer codes}) with parameters $[N,K_i,d(\mathcal{C}_i)]$ over $\mathbb{F}_{q^{k_i}},$ where $\mathbb{F}_{q^{k_i}}$ are extensions of degree $k_i$ of the finite field $\mathbb{F}_q.$ Consider the set $\mathcal{C}$ of $s\times N$- matrices defined as follows

        \begin{equation}\label{array}
        \mathcal{C}:=\left\{c=\left(\begin{array}{ccc} c^{1}_1 & \ldots & c^{1}_N \\\vdots  & \ldots & \vdots\\ c^{s}_1 & \ldots & c^{s}_N \end{array} \right): \ (c^{i}_1, \ldots , c^{i}_N)\in \mathcal{C}_i \ \mbox{for $1\leq i \leq s$} \right\}.
        \end{equation}
        For $k_1+\ldots+k_s\leq n,$ suppose that $\pi: \mathbb{F}_{q^{k_1}}\times \cdots \times \mathbb{F}_{q^{k_s}} \mapsto {\mathbb{F}_q}^n $ is an $\mathbb{F}_q$-linear injection whose image $A:={\rm im}(\pi)$ is a linear code (called an \emph{inner code}) of parameters $[n,{\sum_{i=1}^{s}k_i},d(A)].$ Then the set
        \begin{equation}\label{GCC}
        \pi(\mathcal{C})= \left\{\left( \pi(c_1),\dots,\pi(c_N) \right): c_j\text{'s are columns of } c\in \mathcal{C} \text{, for } j=1,\ldots, N \right\},
        \end{equation}
        is called a \emph{generalized concatenated code $($GCC$)$}.
    \end{defn}

    This concatenation of an inner code $A$ with an outer code $\mathcal{C}$ is denoted throughout by $A \square \mathcal{C}.$ Let us note that {\em simple concatenation} is obtained if we choose a GCC with only one outer code.

    With the following statement, we observe the relation between a GCC with at least two outer codes and a simple concatenation, and we give the parameters of a GCC.

    \begin{thm}\label{GCC_Dumer}
        Let $\pi(\mathcal{C})$ be a GCC as described above. Then the following conditions hold:
        \begin{itemize}
            \item[(i)] A generalized concatenated code $\pi(\mathcal{C})$ is a linear code of length $nN$, of dimension $\displaystyle{\sum_{i=1}^{s}k_{i}K_{i}}$ over $\mathbb{F}_q.$
            \item[(ii)] A generalized concatenated code $\pi(\mathcal{C})$ can be written as a direct sum of simple concatenations. Namely,
            \begin{equation}\label{GCC 2}
            \pi(\mathcal{C})= A_{1} \Box \mathcal{C}_{1} \oplus \cdots \oplus A_{s} \Box \mathcal{C}_{s}= \bigoplus_{i=1}^{s} A_i \Box \mathcal{C}_i,
            \end{equation}
            where $A_{i}=\pi(0,\ldots,0,x_i,0,\ldots,0)$'s are $k_i$-dimensional subcodes of $A$ and $x_i \in \mathbb{F}_{q^{k_i}}.$ Here, $A_i$'s are also called inner codes in the concatenation (\ref{GCC}).
            \item[(iii)]  Conversely, let $A_{i}$'s be $q$-ary linear codes of parameters $[n, k_{i},d(A_{i})]$ with $\displaystyle{A_{j}\cap \sum_{i \neq j} A_{i}=\{0\}},$ and let $\mathcal{C}_{i}$'s be $\mathbb{F}_{q^{k_i}}$-linear codes with parameters $[N, K_{i},d(\mathcal{C}_i)],$ for each $i \in \{1,\ldots,s\}.$ Then the direct sum of simple concatenations $\displaystyle{\bigoplus_{i=1}^{s} A_i \Box \mathcal{C}_i}$ can be redescribed as a \emph{GCC} code.
            \item [(iv)] For a given GCC in the form of $\pi(\mathcal{C})=A_{1} \Box \mathcal{C}_{1} \oplus \cdots \oplus A_{s} \Box \mathcal{C}_{s},$ we have
            $$d\left(\pi(\mathcal{C})\right) \geq \displaystyle{\min_{1\leq i \leq s}} \left\{d(\mathcal{C}_{i})d(A_1  \oplus \cdots \oplus A_i) \right\},$$ if $d\left(\mathcal{C}_{1}\right)\leq d\left(\mathcal{C}_{2}\right) \leq \cdots \leq d\left(\mathcal{C}_{s}\right).$

        \end{itemize}
    \end{thm}

    \medskip

    \subsection{Background on Quasi-Abelian Codes}

    We review the structure of quasi-abelian codes in this section, following \cite{js} closely (see also \cite{dkl}). We refer the reader to these articles for further details. Let us note that in the special case of quasi-cyclic codes, the material presented in this section has analogies with that presented in \cite{ls}.

    Let $G$ be a finite (additive) abelian group of order $n$. Consider the group algebra $\mathbb{F}_q[G]$, whose elements are of the form $\sum_{g\in G} \alpha_{g}Y^g$ for $\alpha_{g}\in \mathbb{F}_{q}$. The multiplicative identity of $\mathbb{F}_q[G]$ is $Y^0$. Note that $\mathbb{F}_q[G]$ can be considered as a vector space over $\mathbb{F}_q$ of dimension $\vert G\vert$.

    We call $\mathcal{C}$ a linear code in $\mathbb{F}_q[G]$ of length $n$ if it is an $\mathbb{F}_q$-subspace of $\mathbb{F}_q[G]$. Note that such a code can be viewed as a linear code of length $n$ over $\mathbb{F}_q$ by indexing the symbols in codewords with the elements in $G$. Hence, the Hamming weight ${\rm wt}(v)$ of $v=\sum_{g \in G} v_{g}Y^g \in \mathbb{F}_q[G]$ is the number of nonzero terms $v_{g}$ and the minimum distance of $\mathcal{C}$ is
    \begin{equation*}\label{minimumdistancebnd}
    d(\mathcal{C}):= \min \{ {\rm wt}(v) \vert v \in \mathcal{C}, v\neq 0 \}.
    \end{equation*}

    \begin{defn}
        A code $\mathcal{C}$ in $\mathbb{F}_q[G]$ is called an $H$ \emph{quasi-abelian code} ($H$-QA) of index $\ell$ if $\mathcal{C}$ is an $\mathbb{F}_q[H]$-module, where $H$ is a subgroup of $G$ with $[G:H]=\ell$.  We will only refer to these codes as QA codes, unless it is needed to specify the subgroup $H$ and the index.
    \end{defn}

    Let $\{g_{1},\ldots,g_{\ell}\}$ be a fixed set of representatives of the cosets of $H$ in $G$. Note that a QA code of index $\ell$ in $\mathbb{F}_q[G]$ can be seen as an $\F_q[H]$-submodule of $\F_q[H]^{\ell}$ by the following $\F_q[H]$-module isomorphism.
    \begin{equation}\begin{array}{lclc} \label{identification-1}
    \Phi: & \mathbb{F}_q[G] &\longrightarrow & \F_q[H]^\ell \\
    & \displaystyle{\sum_{i=1}^{\ell}\sum_{h \in H} \alpha_{h+g_{i}}Y^{h+g_{i}}}
    & \longmapsto & \displaystyle{\left(\sum_{h\in H} \alpha_{h+g_{1}}Y^h, \dots, \sum_{h\in H} \alpha_{h+g_{\ell}}Y^h \right)}.
    \end{array}
    \end{equation}

    \medskip

    \begin{rem}\label{special cases}

        It is clear that an $H$-QA code is QC if $H$ is cyclic. Moreover, if $H=J\times K$ with $J$ cyclic and $|K|=t$, then an $H$-QA code of index $\ell$ is a QC code of index $t\ell.$ By the fundamental theorem of finite abelian groups, every abelian group $H$ decomposes into products of cyclic groups. Hence the class of QA codes is a subclass of QC codes. For instance, if we choose $H=C_{m_1}\times C_{m_2},$ where $C_{m_i}$ denotes the cyclic group $\mathbb{Z}/{m_{i}}\mathbb{Z}$ of order $m_i$ for $i=1,2$, then an $H$-QA code can be viewed as a QC code of co-index $m_1$ or co-index $m_2.$ Moreover, as mentioned before in \cite{goz} and \cite{j} for certain special cases, we have various QA structures with different indices for a given QA code, since an $\mathbb{F}_q[H]$-module in $\mathbb{F}_q[H]^{\ell}$ is also an $\mathbb{F}_q[H^{'}]$-module, for any $H^{'}\leq H \leq G$.

        Jitman and Ling (\cite{js}) call a QA code $\mathcal{C}$ strictly QA (SQA) if $H$ is not a cyclic group. Similarly, if $\ell=1$ and $H$ is not cyclic, we refer to strictly abelian (SA) codes. 

    \end{rem}

    We continue with recalling the CRT decomposition of $H$-QA codes of index $\ell,$ which is introduced in \cite{js} (see also \cite{dkl}). For a semisimple algebra $\mathbb{F}_q[H],$ where $H$ is a subgroup of a finite abelian $G$ with $|H|=m,$ let $M$ be the exponent of $H$ and let  $\mathbb{K}$ be an extension of $\mathbb{F}_q$ which contains a primitive $M$-th root of unity $\xi$. We set $R:=\mathbb{F}_q[H]$ throughout.

    A character $\chi$ from $H$ to the multiplicative group of $\K$ is a group homomorphism. The set $\Hom(H,\K^{*})$ of characters forms a group  which is isomorphic to $H$. So, we can denote the characters in $\Hom(H,\K^{*})$ as $\chi_{a}$, $a \in H$. If we view the abelian group $H$ as a direct product of finite cyclic groups,
    \begin{equation*}\label{decompfiniteabln}
    H=\prod_{i=1}^{s} C_{m_i},
    \end{equation*}
    then an element $h\in H$ can be represented as $h=(h_{1}, \cdots, h_{s})$, where $h_{i}\in C_{m_i}$ and $C_{m_i}$ denotes the additive cyclic group $ \mathbb{Z}/ {m_{i}}\mathbb{Z}$ of order $m_i.$ In this case, it is well-known that the character $\chi_{a}$ can be written as \begin{equation}\label{char}
    \chi_{a}(h)=\xi^{\sum_{i=1}^{s}a_{i}h_{i}M/m_{i}},
    \end{equation} for any $a\in H$.

    Recall that a primitive idempotent of a ring is a nonzero element $e$ such that $e^2=e$ and for any other idempotent $f$, either $ef=0$ or $ef=e$. To present the decomposition of QA codes, we will need to use idempotents in $R=\F_q[H]$. For this purpose, one first considers the group algebra $\mathbb{K}[H]$, whose primitive idempotents are given by
    \begin{equation}\label{idempotentsinK}
    E_{x}= \frac{1}{m}\sum_{a\in H}\chi_{x}(-a)Y^{a} \in \mathbb{K}[H],
    \end{equation}
    for each $x\in H$. The primitive idempotents of $\mathbb{K}[H]$ are orthogonal, i.e. $E_xE_y=0$ if $x,y\in H$ and $x\not= y$.

    The $q$-\emph{cyclotomic class} of $H$ containing $h\in H$ is defined as
    \begin{equation}\label{qcyccosets}
    S_{q}(h):= \{q^{i}h : 0 \leq i < v_{h}\},
    \end{equation}
    where $q^{i}h$ denotes addition of $h$ with itself $q^i$ times (recall that $G$ and hence $H$ are additive groups), and $v_h$ is the smallest positive integer such that $q^{v_h}\equiv 1  \ ({\rm ord} \ h)$.
    Primitive idempotents in $\mathbb{F}_q[H]$ are of the form
    \begin{equation}\label{idemptsinF}
    \displaystyle e_h= \sum_{x\in S_{q}(h)}E_x,
    \end{equation}
    where $h\in H$ and $E_x$ is a primitive idempotent in $\mathbb{K}[H]$ as in (\ref{idempotentsinK}). The idempotent $e_h$ is called the primitive idempotent induced by $S_{q}(h)$. Orthogonality of the primitive idempotents of $\mathbb{K}[H]$ implies orthogonality of the primitive idempotents in $\mathbb{F}_q[H]$:
    \begin{equation}\label{orthogonality}
    e_he_{h'}=0, \ \ \mbox{if $h,h'\in H$ have distinct $q$-cyclotomic classes.}
    \end{equation}
    If $S_{q}(h_1), \ldots, S_{q}(h_t)$ are all $q$-cyclotomic classes of $H$ and $e_{h_1},\ldots,e_{h_t}$ are the corresponding primitive idempotents of $\mathbb{F}_q[H]$, then we have
    \begin{equation}\label{idempsum}
    \sum_{i=1}^{t}e_{h_i}=1.
    \end{equation}
    Moreover, primitive idempotents of $R=\F_q[H]$ yields the decomposition
    \begin{equation}\label{decompR}
    R=\displaystyle{\bigoplus_{i=1}^{t}Re_{h_i}}.
    \end{equation}
    The ideal $Re_{h_i}$ generated by $e_{h_i}$ in the group algebra $R$ is an abelian code (\cite{j}). Moreover, $Re_{h_i}$ is an extension field of $\mathbb{F}_q$ with the extension degree $S_{q}(h_i)$ (for all $1\leq i \leq t$). The maps yielding the identification of $Re_{h_i}$ with the extension $\mathbb{E}_i$ of $\F_q$ are
    \begin{equation}\begin{array}{lcll} \label{concatenation map-1}
    \varphi_i: &  Re_{h_i}&\longrightarrow & \mathbb{E}_i \\
    &\displaystyle{\left(\sum_{h\in H}\alpha_{h}Y^{h}\right)e_{h_i}}
    &\longmapsto & \displaystyle{\sum_{h\in H}\alpha_{h}\chi_{h_i}(h)},
    \end{array}
    \end{equation}

    \begin{equation}\begin{array}{lclc} \label{concatenation map-2}
    \psi_i: & \mathbb{E}_i &\longrightarrow & Re_{h_i}  \\
    & \delta &\longmapsto & \displaystyle{\sum_{k\in H}\alpha_{k}Y^{k}},
    \end{array}
    \end{equation}
    where $\alpha_{k}=\frac{1}{m}\Tr(\delta  \chi_{h_i}(-k))$. Here, $\Tr$ denotes the trace map from $\mathbb{E}_i$ to $\mathbb{F}_q$. Note that $\varphi_i$ and $\psi_i$ are nontrivial ring homomorphisms and they are inverse to each other for every $1\leq i \leq t.$ Moreover, $\varphi_i(e_{h_i})=1$ and hence $\psi_i(1)=e_{h_i}$.

    By (\ref{idempsum}), any element $r\in R$ can be written as $r=re_{h_1}+\cdots +re_{h_t}$. For an element $(r_1,\ldots ,r_\ell)\in R^\ell$, we have
    \begin{eqnarray*}
        (r_1,\ldots ,r_\ell)&=&(r_1e_{h_1}+\cdots +r_1e_{h_t}, \ldots ,r_{\ell}e_{h_1}+\cdots +r_{\ell}e_{h_t} )\\
        &=& (r_1e_{h_1},\ldots ,r_\ell e_{h_1})+\cdots + (r_1e_{h_t},\ldots ,r_\ell e_{h_t}).
    \end{eqnarray*}
    Using the isomorphisms $\varphi_1 ,\ldots ,\varphi_t$, we can identify $R^\ell$ and $\oplus_{i=1}^{t} \mathbb{E}_i^{\ell}$:
    \[\begin{array}{ccc}
    R^\ell & \longrightarrow & \mathbb{E}_1^{\ell}\oplus \cdots \oplus \mathbb{E}_t^{\ell} \\
    (r_1,\ldots ,r_\ell) & \longmapsto & \bigl(\varphi_1(r_1e_{h_1}),\ldots ,\varphi_1(r_\ell e_{h_1})\bigr)+\cdots + \bigl(\varphi_t(r_1e_{h_t}),\ldots ,\varphi_t(r_\ell e_{h_t})\bigr)
    \end{array}\]
    Consequently, an $R$-submodule of $R^{\ell}$ can be viewed as $\oplus_{i=1}^{t} \mathbb{E}_i$-submodule of $\oplus_{i=1}^{t} \mathbb{E}_i^{\ell}.$ Therefore, a QA code $C\subseteq R^\ell$ decomposes as
    \begin{equation} \label{constituents}
    \mathcal{C}=\mathcal{C}_1\oplus \cdots  \oplus \mathcal{C}_{t},
    \end{equation}
    where $\mathcal{C}_i\subset \mathbb{E}_i^\ell$ is a linear code of length $\ell$ over the field $\mathbb{E}_i$ for every $1 \leq i \leq t$. We call $\mathcal{C}_i$'s the constituents of $\mathcal{C}$. The preceding arguments yield the explicit description of the constituents (for $1\leq i \leq t$):
    \begin{equation}\label{explicitconst}
    \mathcal{C}_i=\Bigl\{\bigl(\varphi_i(c_1e_{h_i}),\ldots ,\varphi_i(c_\ell e_{h_i})  \bigr): (c_1,\ldots ,c_\ell)\in C \Bigr\}.
    \end{equation}

    \medskip

    \section{The Concatenated Structure of Quasi-Abelian Codes}\label{results}

    Jensen gave the concatenated structure of abelian and QC codes in \cite{j}. It was later shown in \cite{go} that the CRT decomposition of a QC code in \cite{ls} and the code's concatenated decomposition by Jensen are equivalent. Here, we give the concatenated structure of QA codes and prove the analog of the result in \cite{go}.

    Consider the rings $R^\ell=\mathbb{F}_q[H]^\ell$ and $\mathbb{E}_{i}^{\ell}$ (for $i\in\{1,\ldots,t\}$), where the ring operations are clearly componentwise addition and multiplication. Using the maps $\varphi_i$ and $\psi_i$ in \eqref{concatenation map-1} and \eqref{concatenation map-2}, we define
    \begin{equation}\begin{array}{lclc} \label{concatenation map-3}
    \Psi_i: & \mathbb{E}_i^\ell &\longrightarrow & R^\ell  \\
    & (a_1,\ldots,a_\ell)
    &\longmapsto & (\psi_{i}(a_1),\ldots,\psi_{i}(a_\ell))
    \end{array}
    \end{equation}
    and
    \begin{equation}\begin{array}{lclc} \label{concatenation map-4}
    \Phi_i: &  R^\ell  &\longrightarrow & \mathbb{E}_i^\ell  \\
    & \left(\displaystyle{\sum_{h\in H} \alpha_h^{1}Y^h,\ldots,\sum_{h\in H} \alpha_h^{\ell}Y^h} \right) &\longmapsto & \left(\displaystyle{\sum_{h\in H}\alpha_h^{1}\chi_{h_i}(h),\ldots,\sum_{h\in H} \alpha_h^{\ell}\chi_{h_i}(h)}\right).
    \end{array}
    \end{equation}
    Note that $\Psi_i$ and $\Phi_i$ are $\F_q$-linear ring homomorphisms (for $i\in\{1,\ldots,t\}$). Moreover they are inverse to each other when $\Phi_i$ is restricted to the image of $\Psi_i$. Next we describe the primitive idempotents of  $R^\ell$.

    \begin{thm} \label{idempts. in FG^l}
        For each $i\in\{1,\ldots,t\}$, let $\Theta_{i}:= \Psi_{i}(1,\ldots,1)=(e_{h_i},\ldots, e_{h_i})$. Then $\langle \Theta_i \rangle = \Psi_{i}(\mathbb{E}_{i}^\ell)$ and $\displaystyle R^\ell=\bigoplus_{i=1}^{t} \langle \Theta_{i} \rangle $. Moreover,
        \begin{equation*}\label{idempotent}
        \Theta_{i} \Theta_{j}=
        \begin{cases}
        \Theta_{i} & \text{if } i = j,\\
        0 & \text{if } i \neq j,
        \end{cases}
        \end{equation*} and
        $\displaystyle \sum_{i=1}^{t} \Theta_{i}= (1,\ldots,1) $ in $R^\ell.$
    \end{thm}
    \begin{proof}
        The equality $\langle \Theta_i \rangle = \Psi_{i}(\mathbb{E}_{i}^\ell)$ follows immediately from the definitions of $\psi_i$ and $\Psi_i$. Suppose $(f_1,\ldots ,f_\ell)$ in $R^\ell$ belongs to the intersection of $\langle \Theta_i \rangle$ and $\langle \Theta_j \rangle$ for $i\not= j$. This implies that for all $u\in\{1,\ldots,\ell\}$, $f_u\in Re_{h_i}\cap Re_{h_j}$, which is trivial by (\ref{decompR}). So, $ \displaystyle \bigoplus_{i=1}^{t} \langle \Theta_{i} \rangle $ is indeed a direct sum in $R^\ell$. Since $\langle \Theta_i \rangle = \Psi_{i}(\mathbb{E}_{i}^\ell)$, we have $\dim_{\F_q}\langle \Theta_i \rangle=\ell [\E_i:\F_q]$. Hence,
        \begin{eqnarray*}
            \dim_{\F_q} \displaystyle \bigoplus_{i=1}^{t} \langle \Theta_{i} \rangle & = & \ell \sum_{i=1}^t [\E_i:\F_q]\\
            &=&\ell \sum_{i=1}^t \dim_{\F_q}Re_{h_i} \ \ \mbox{by (\ref{concatenation map-1})}\\
            &=& \ell \dim_{\F_q}R \ \ \mbox{by (\ref{decompR})}.
        \end{eqnarray*}
        Hence, $R^\ell=\displaystyle \bigoplus_{i=1}^{t} \langle \Theta_{i} \rangle$. The other assertions easily follow from (\ref{orthogonality}) and (\ref{idempsum}).
    \end{proof}

    Next, we describe the concatenated structure of QA codes. We denote the concatenation operation by $\Box$ as commonly done in the literature. In the following, we use the set defined as
    $$\mathcal{C} s := \{c s : \ c\in \mathcal{C}\},$$
    for $\mathcal{C} \subseteq R^\ell$ and an element $s\in R^\ell$.

    \begin{thm} \label{concatenated decomposition}
        With the notation above, the following conditions hold:
        \begin{itemize}
            \item [ (i)] Let $\mathcal{C}$ be an $R$-submodule of $R^\ell$ and $\tilde{\mathcal{C}}_i:= \mathcal{C}  \Theta_i \subseteq R^\ell$ for all $i=1,\ldots, t$. Then, for some subset $\mathfrak{I} \subseteq \{1,\ldots,t\}$, we have $\mathcal{C}=\bigoplus_{i\in \mathfrak{I}} \tilde{\mathcal{C}}_i$. Moreover, $\tilde{\mathcal{C}}_i= R e_{h_i}  \Box \mathfrak{C}_i $, where $\mathfrak{C}_i=\Phi_i(\tilde{\mathcal{C}_i})$ is an $\mathbb{E}_i$-linear code of length $\ell$ for each $i$.
            \item [(ii)] Conversely, let $\mathfrak{C}_i$ be a linear code over $\mathbb{E}_i$ of length $\ell$ for all $i$ in some subset $\mathfrak{I}$ of $\{1, \ldots,t\}$. Then, $\displaystyle \mathcal{C}= \bigoplus_{i\in \mathfrak{I}} R e_{h_i} \Box \mathfrak{C}_i$ is an $H$-QA code of index $\ell$.
        \end{itemize}
    \end{thm}

    \begin{proof}
        (i) By Theorem \ref{idempts. in FG^l} we have
        \begin{equation*}
        \displaystyle \mathcal{C}=\mathcal{C} \sum_{i=1}^{t}\Theta_{i}= \sum_{i\in \mathfrak{I}} \tilde{\mathcal{C}_i},
        \end{equation*}
        where $\mathfrak{I}$ consists of indices $i$ for which $\tilde{\mathcal{C}_i}\neq \{0\}.$ Since $\tilde{\mathcal{C}_i}$ lies in the ideal $\langle \Theta_{i} \rangle$ and the sum of these ideals is direct, the sum $\displaystyle \sum_{i\in \mathfrak{I}}\tilde{\mathcal{C}_i}$ is also direct.

        On the other hand, for all $i\in \mathfrak{I}$, we have
        \begin{eqnarray*}
            \tilde{\mathcal{C}}_i & = & \mathcal{C} \Theta_{i}\\
            & = & \left\{\left(c_{1},\ldots,c_{\ell}\right) \left(e_{h_i},\ldots, e_{h_i}\right): \left(c_1,\ldots,c_{\ell}\right)\in \mathcal{C}\right\}\\
            & = & \{(c_{1}e_{h_i},\ldots,c_{\ell}e_{h_i}): (c_1,\ldots,c_{\ell})\in \mathcal{C}\}.
        \end{eqnarray*}
        Hence,
        \begin{equation}\label{eq-2}
        \mathfrak{C}_i = \Phi_{i}\left(\tilde{\mathcal{C}_i}\right)= \Bigl\{ \bigl(\varphi_{i}(c_{1}e_{h_i}),\ldots,\varphi_{i}(c_{\ell}e_{h_i})\bigr): (c_1,\ldots,c_{\ell})\in \mathcal{C} \Bigr\}.
        \end{equation}

        Since $\tilde{\mathcal{C}_i}$ and $\Phi_{i}$ are $\mathbb{F}_q$-linear, each $\mathfrak{C}_i$ is an $\mathbb{F}_q$-linear code of length $\ell$. The map $\varphi_i$ in (\ref{concatenation map-1}) is bijective. Therefore for any $\delta \in \mathbb{E}_i$, there exists $f\in R$ such that $\varphi_{i}(fe_{h_i})= \delta.$ So, for any $(\varphi_{i}(c_{1}e_{h_i}),\dots,\varphi_{i}(c_{\ell}e_{h_i})) \in \mathfrak{C}_i $, we have
        \begin{eqnarray}
        \delta (\varphi_{i}(c_{1}e_{h_i}),\dots,\varphi_{i}(c_{\ell}e_{h_i})) &=& (\varphi_{i}(fe_{h_i})\varphi_{i}(c_{1}e_{h_i}),\dots,\varphi_{i}(fe_{h_i})\varphi_{i}(c_{\ell}e_{h_i})) \nonumber\\
        &=& (\varphi_{i}(fc_{1}e_{h_i}),\dots,\varphi_{i}(fc_{\ell}e_{h_i})).   \label{eq-1}
        \end{eqnarray}
        Since $\mathcal{C}$ is an $R$-module, $(fc_{1},\dots,fc_{\ell})$ lies in $\mathcal{C}$. Therefore, (\ref{eq-1}) belongs to $\mathfrak{C}_i$, which shows that $\mathfrak{C}_i$ is $\mathbb{E}_i$-linear.

        Now, consider the concatenated code $R e_{h_i}  \Box \mathfrak{C}_i$ determined by $\psi_{i}: \mathbb{E}_i \rightarrow Re_{h_i}$ in (\ref{concatenation map-2}):
        $$R e_{h_i} \Box \mathfrak{C}_{i} = \left\{\left( \psi_i(\varphi_i(c_1e_{h_i})),\ldots,\psi_i(\varphi_i(c_{\ell}e_{h_i}))  \right): (c_{1},\ldots,c_{\ell})\in \mathcal{C} \right\}.$$
        Since $\psi_i$ and $\phi_i$ are inverse to each other, we have
        $$R e_{h_i} \Box \mathfrak{C}_{i} = \left\{( c_{1}e_{h_i},\ldots, c_{\ell}e_{h_i}): (c_{1},\ldots,c_{\ell})\in \mathcal{C} \right\}=\tilde{\mathcal{C}_i},$$
        which completes the proof.

        (ii) Let $\mathfrak{C}_i$ be an $\mathbb{E}_i$ linear code of length $\ell$ and consider the concatenation
        $$
        Re_{h_i} \Box \mathfrak{C}_{i}= \{( \psi_{i}(\lambda_{1}),\ldots, \psi_{i}(\lambda_{\ell})) : (\lambda_{1},\ldots,\lambda_{\ell})\in \mathfrak{C}_i \}
        $$
        for each $i\in \mathfrak{I}$. By linearity of $\mathfrak{C_i}$ and $\psi_i$, this set becomes an additive subgroup of $R^{\ell}$. We need to show that $Re_{h_i} \Box \mathfrak{C}_{i}$ is closed under multiplication by elements of $R$. For this, it is enough to show that it is closed under multiplication by $Y^x\in R$, for any $x\in H$. Since $\varphi_i$ is surjective, we can write an element $(\lambda_{1},\ldots ,\lambda_{\ell})\in \mathfrak{C}_i$ as
        $(\varphi_i(f_1e_{h_i}),\ldots ,\varphi_i(f_\ell e_{h_i}))$ for some $f_1,\ldots ,f_\ell \in R$. Then,
        \begin{eqnarray}
        Y^x( \psi_{i}(\lambda_{1}),\ldots, \psi_{i}(\lambda_{\ell})) &=&Y^x(f_1e_{h_i},\ldots ,f_\ell e_{h_i}) \ \ \ \mbox{($\psi_i$ and $\varphi_i$ are inverse)} \nonumber \\
        &=& \bigl(Y^xe_{h_i}f_1e_{h_i},\ldots ,Y^xe_{h_i}f_\ell e_{h_i}\bigr) \ \ \ \mbox{(using $e_{h_i}e_{h_i}=e_{h_i}$)} \nonumber\\
        &=& \left(\psi_{i}\Bigl(\varphi_i \bigl(Y^xe_{h_i}f_1e_{h_i}\bigr)\Bigr),\ldots ,\psi_{i}\Bigl(\varphi_i \bigl(Y^xe_{h_i}f_\ell e_{h_i}\bigr)\Bigr)\right) \nonumber\\
        &=& \left(\psi_{i}\Bigl(\varphi_i \bigl(Y^xe_{h_i}\bigr) \varphi_i \bigl(f_1e_{h_i}\bigr)\Bigr),\ldots ,\psi_{i}\Bigl(\varphi_i \bigl(Y^xe_{h_i}\bigr) \varphi_i \bigl(f_\ell e_{h_i}\bigr)\Bigr)\right). \label{eq-3}
        \end{eqnarray}
        Since $\varphi_i \bigl(Y^xe_{h_i}\bigr)$ is in $\mathbb{E}_i$ and $\mathfrak{C}_i$ is $\mathbb{E}_i$-linear, $\Bigl(\varphi_i \bigl(Y^xe_{h_i}\bigr) \varphi_i \bigl(f_1e_{h_i}\bigr),\ldots , \varphi_i \bigl(Y^xe_{h_i}\bigr) \varphi_i \bigl(f_\ell e_{h_i}\bigr) \Bigr)$ belongs to $\mathfrak{C}_i$. Therefore (\ref{eq-3}) is in $Re_{h_i}\Box \mathfrak{C}_i$.

        Finally, $Re_{h_i} \Box \mathfrak{C}_{i}$ lies in $(Re_{h_i})^\ell$ (for each $i$) and $Re_{h_i}$'s intersect trivially (cf. (\ref{decompR})). Therefore the sum of the concatenations $Re_{h_i} \Box \mathfrak{C}_{i}$, for $i\in \mathfrak{I}$, is direct. Hence the result follows.
    \end{proof}

    \begin{rem}
        For a QA code $\mathcal{C}$, the constituent $\mathcal{C}_i$ and the outer code $\mathfrak{C}_i$ in its concatenated form coincide, for each $i$. This follows from (\ref{explicitconst}) and (\ref{eq-2}).
    \end{rem}

    The minimum distance bound which is valid for all concatenated codes (see \cite{bz, d}, and part $iv$ in Theorem \ref{GCC_Dumer},) apply to QA codes by Theorem \ref{concatenated decomposition}. So, we do not prove the next result.

    \begin{cor} \label{minimum distance bound}
        Let $\mathcal{C}$ be a QA code of index $\ell$ in $R^{\ell}$ with the concatenated structure
        $$C = \bigoplus_{j=1}^g R e_{h_{i_j}}  \Box \mathfrak{C}_{i_j},$$
        where $\mathfrak{C}_{i_j}$'s are the nonzero outer codes (constituents) of $\mathcal{C},$ and $Re_{h_{i_j}}$'s are minimal abelian codes generated by primitive idempotents in $R^{\ell}$ for $\{i_1,\ldots ,i_g\} \subseteq \{1,\ldots,t\}.$ Assume that $d\left(\mathfrak{C}_{i_1}\right)\leq d\left(\mathfrak{C}_{i_2}\right) \leq \cdots \leq d\left(\mathfrak{C}_{i_g}\right)$. Then, we have
        $$d\left(C\right) \geq \displaystyle{\min_{1\leq v \leq g}} \left\{d(\mathfrak{C}_{i_v})d(Re_{h_{i_1}}  \oplus \cdots \oplus R e_{h_{i_v}}) \right\}.$$
    \end{cor}

    \medskip

    \section{Numerical Results and Asymptotics}\label{results2}
    The concatenated structure of QA codes in Section \ref{results} allows us to obtain numerical results and asymptotic conclusions for QA codes, as it is shown here.

    \subsection[Numerical Results]{Numerical Results}

    Corollary \ref{minimum distance bound} gives a theoretical bound for the minimum distance of QA codes. As we have already observed, the class of SQA codes forms a subfamily of the class of QC and this last class contains many optimal codes. In this part, we present some numerical results to see how good the bound obtained is and to show the existence of some optimal SQA codes, different from those found by Jitman and Ling in \cite{js}. In the first three examples, particularly, we show some examples of optimal SQA codes of index $3$ and $4$, while in the last two examples we present certain SQA codes of rate close to $1/2$.

    We develop an algorithm in {\sc Magma} to compute, given a finite group $H$, all $H$-QA codes of a fixed index $\ell$ with a minimum distance bounded below by certain constant $d$. The algorithm follows the following steps:
    \begin{itemize}
        \item compute all $q$-cyclotomic classes $S_{q}(h)$ in $\mathbb{F}_{q}[H]$, whose cardinalities give the degrees of the extension fields;
        \item for each extension $\mathbb{F}_{q^k}$, where $k=|S_{q}(h)|,$ compute all linear codes of length $\ell$ over $\mathbb{F}_{q^k}$;
        \item for each element $h$ in a $q$-cyclotomic class $S_{q}(h)$
        and for each linear code $\mathcal{C}$ of length $\ell$ over $\mathbb{F}_{q^k}$, by using DFT, which is written explicitly in \ref{concatenation map-2} and chosen as a concatenation map, we compute the concatenation of $\mathcal{C}$ with the minimal abelian code given by $h$. We put the obtained codes in a list $S$ if their minimum distance is greater than or equal to $d$;
        \item we sum pairs of elements of $S$, always checking if their minimum distance is greater than or equal to $d$, and we put the obtained codes in a list $S'$. We repeat the process to obtain codes with a higher dimension.
    \end{itemize}

    The complexity of this algorithm strongly depends on $\ell$, on the size of the $q$-cyclotomic classes and on the chosen $d$.

    We applied this algorithm in some cases and obtained some optimal codes. As an example, we first give here three cases, of index $2$, $3$ and $4$ respectively, for which the codes obtained have the best-known minimum distance for their dimension. Moreover, in the second and the third, the minimum distance obtained meets the lower bound given by Corollary \ref{minimum distance bound}.


    We looked for QA codes in $R:=\mathbb{F}_2[C_5\times C_5]^{2}$ of minimum distance at least $18$ and, for the dimension $12$, we got only one $[50,12,18]$ code $\mathcal{C}$, up to equivalence, which is the direct sum
    $$\mathcal{C}=(Re_{(1,0)}
    \Box \mathcal{C}_1)\oplus
    (Re_{(1,1)}\Box \mathcal{C}_1)\oplus
    (Re_{(2,4)}\Box \mathcal{C}_2)$$
    with $\mathcal{C}_1$ of generator matrix $G_1:=[1,\alpha^7]$ and $\mathcal{C}_2$ of generator matrix $G_2:=[1,\alpha^{12}]$, where $\alpha$ is a primitive element in $\mathbb{F}_{16}$ such that $\alpha^4=\alpha+1$.\\

    Note that $18$ is the best known minimum distance for a code of length $50$ and dimension $12$, according to Grassl's tables \cite{ct}. Moreover, in this case the lower bound given by Corollary \ref{minimum distance bound} is $12$.

    Secondly, we looked for QA codes in $R:=\mathbb{F}_2[C_3\times C_3]^{3}$ of minimum distance at least $12$ and, for the dimension $6$, we got only one $[27,6,12]$ code $\mathcal{C}$, up to equivalence, with two outer codes, namely
    $$\mathcal{C}=(Re_{(2,2)}
    \Box \mathcal{C}_1)\oplus
    (Re_{(1,0)}\Box \mathcal{C}_2)$$
    with $\mathcal{C}_1$ of generator matrix $G_1:=\left[\begin{smallmatrix}1&0&1\\0&1&\alpha\end{smallmatrix}\right]$ and $\mathcal{C}_2$ of generator matrix $G_2:=[1,\alpha,1]$, where $\alpha$ is a primitive element in $\mathbb{F}_{4}$.\\
    Note that $12$ is the best known minimum distance for a code of length $27$ and dimension $6$, by the Griesmer bound \cite{ct}. Moreover, in this case the lower bound given by Corollary \ref{minimum distance bound} is exactly $12$.

    Now, we looked for QA codes in $R:=\mathbb{F}_2[C_3\times C_3]^{4}$ of minimum distance at least $16$ and, for the dimension $6$, we got only one $[36,6,16]$ code $\mathcal{C}$, up to equivalence, with two outer codes, namely
    $$\mathcal{C}=(Re_{(2,2)}
    \Box \mathcal{C}_1)\oplus
    (Re_{(1,0)}\Box \mathcal{C}_2)$$
    with $\mathcal{C}_1$ of generator matrix $G_1:=\left[\begin{smallmatrix}1&0&\alpha^2&\alpha\\0&1&1&\alpha\end{smallmatrix}\right]$ and $\mathcal{C}_2$ of generator matrix $G_2:=[1,\alpha,\alpha,\alpha]$, where $\alpha$ is a primitive element in $\mathbb{F}_{4}$.\\
    Note that $16$ is the best known minimum distance for a code of length $36$ and dimension $6$, by the Griesmer bound \cite{ct}. Moreover, in this case the lower bound given by Corollary \ref{minimum distance bound} is exactly $16$.

        In the last two examples, in order to show the effectiveness the concatenation method for QA codes, we consider a binary QA in $\mathbb{F}_{2}[C_{5}\times C_{5}]^{256}$ of minimum distance at least $48$ and of rate $\approx 1/2$, with the concatenated structure

        $$\mathcal{C}= \bigoplus_{i=1}^{4}(Re_{h_i}
        \Box \mathcal{C}_i),$$

        where all $\mathcal{C}_{i}$ 's are Reed-Muller codes $\mathcal{RM}_{\mathbb{F}_{16}}(20, 2)$ of parameters $[256, 201, 12]$ over $\mathbb{F}_{16}$ (\cite{ak}), and $h_i$'s are $(1,0),(0,1), (1,1),$ and $(1,2)$ in $H= C_5\times C_5$, respectively. In fact, since each inner code $R_{h_i}$, for the corresponting $h_i$, has parameters $[25,4,10]$ over binary field, and since their sum has minimum distance $4$, the QA code with the given concatenated structure is of parameters $[6400, 3216, \geq 48]$.

    Finally, we consider a QA code in $R:= \mathbb{F}_{3}[C_5\times C_5]^{6561}$ of minimum distance at least $220$ and of rate close to $1/2$, with the concatenated structure

        $$\mathcal{C}=\bigoplus_{i=1}^{4}(Re_{h_i}
    \Box \mathcal{C}_i),$$
    where all $\mathcal{C}_{i}$'s are Reed-Muller codes $\mathcal{RM}_{81}(106, 2)$ of parameters $[6561,5076, 55]$ over $\mathbb{F}_{3^4}$ (\cite{ak}), and $h_i$'s are $(1,0),(1,1),(0,1),$ and $(1,2)$ in $H=C_5\times C_5$, respectively. In fact, using the same arguments as before, we obtain that the  QA code with the given concatenated structure has parameters $[164025, 81216, \geq 220]$.

    \subsection[Asymptotic Results]{Asymptotic Results}

    The class of binary self-dual doubly even strictly QA codes has been shown to be asymptotically good (\cite[Theorem 7.2]{js}), followed by the asymptotical goodness of binary complementary dual QA (QA LCD) codes of index $3$ (\cite{jpc}). Recall that, by a strictly QA code, the authors mean a QA code which is not QC. Note that $H$ being a noncyclic abelian group is enough for this purpose (cf. Remark \ref{special cases}). We first show that strictly QA codes are asymptotically good over any finite field $\F_q$.

    \begin{thm}\label{asymp}
        For any prime power $q$, the class of strictly QA codes over $\F_q$ is asymptotically good.
    \end{thm}

    \begin{proof}
        Let $p$ be a prime different than $\mbox{char}(\F_q)$ and set $H=C_p \times C_p$ so that it is not cyclic and $\gcd(|H|,q)=1$. Note that the $q$-cyclotomic class of $0\in H$ consists of itself only. Let us denote the primitive idempotent corresponding to this cyclotomic class by $e_0$. Hence in the decomposition (\ref{decompR}) of $\F_q[H]$, there exists the field $\F_q$, which is isomorphic to the ideal $\F_q[H]e_0$. This implies that an $H$-QA code over $\F_q$ of any index $\ell$ has a constituent which lies in $\F_q^\ell$.

        Let $\sF:=(\sF_1, \sF_2, \ldots )$ be an asymptotically good family of $\F_q$-linear codes and let the parameters of any member $\sF_i$ in the family be $(n_i,k_i,d_i)$. Define the groups $$G_i:=H\times C_{n_i},$$
        for all $i\geq 1$. We can construct $H$-QA codes $\sE_i$ in $\F_q[G_i]$ (or, in $\F_q[H]^{n_i}$) for all $i$ using Theorem \ref{concatenated decomposition} as follows:
        $$\sE_i:=\F_q[H]e_0 \Box \sF_i.$$
        Note that any member $\sE_i$ of the family $\sE:=(\sE_1,\sE_2,\ldots)$ of $H$-QA codes has parameters $(p^2n_i,k_i, \geq d d_i)$, where $d$ is the minimum distance of the fixed abelian code $\F_q[H]e_0$ of length $p^2$. Hence, the relative parameters of $\sE_i$'s also have positive limit, namely $1/p^2$ multiple of the limit relative rate of $\sF$ and at least $d/p^2$ multiple of the limit relative distance of $\sF$.
    \end{proof}

    We can extend the preceding asymptotic conclusion to the linear complementary dual (LCD) class over any finite field. Let us note that the decomposition of the dual of a QA code is given in \cite{js}. Based on this, a characterization of self-dual QA codes is obtained in terms of the constituents of the code (\cite[Corollary 4.1]{js}). The analogous result for QA LCD codes can be obtained in a straightforward way, so we do not prove it. One can consult \cite[Theorem 3.1]{gos} for the special case of LCD QC codes.

    \begin{thm}\label{lcdasymp}
        For any prime power $q$, the class of strictly QA LCD codes over $\F_q$ is asymptotically good.
    \end{thm}

    \begin{proof}
        Let $H$ be as in the proof of Theorem \ref{asymp} and choose $\sF$ to be an asymptotically good family of LCD codes this time. Such codes exist by \cite{mass} and \cite{sendr}. Consider the family $\sE$ of strictly QA codes as in the same proof. The fact that this family is asymptotically good follows as above. For any $i\geq 1$, the code $\sE_i$ has unique nonzero constituent (namely, $\sF_i$) which is LCD by construction. All other constituents of $\sE_i$ are $\{0\}$, which is trivially LCD with respect to the Euclidean inner product. Hence, by the dual QA code description \cite[p. 519]{js}, and the discussion preceding this theorem, each $\sE_i$ is LCD. Therefore $\sE$ is an asymptotically good family of strictly QA LCD codes.
    \end{proof}

    We note that the codes presented in Theorems \ref{asymp} and \ref{lcdasymp} resemble the asymptotically good QC codes presented in \cite{LS1}, since the ``co-index'' (i.e. length/index) of each code in the families considered is fixed (unlike the asymptotically good family presented in \cite{jpc,js}).

    \medskip

    \section{Acknowledgments} \label{ack}
    The first author was partially supported by PEPS - Jeunes Chercheur-e-s - 2018.
    The second author was supported by T\"{U}B\.{I}TAK Project no. 114F432. The third author was supported by T\"{U}B\.{I}TAK 2214/A Fellowship program.


\end{document}